\documentclass{article}
\usepackage{authblk}

\usepackage{graphicx}
\usepackage{subcaption}

\usepackage{commath}
\usepackage{amsmath}
\usepackage{amssymb}
\usepackage{amsthm}
\usepackage[ruled]{algorithm2e}
\usepackage{algorithmic}

\usepackage[backend=biber,sorting=none]{biblatex}
\newtheorem{theorem}{Theorem}[section]

\newtheorem{lemma}[theorem]{Lemma}
\newtheorem{remark}{Remark}
\newtheorem{assumption}{Assumption}

\DeclareMathOperator{\Tr}{trace}

\usepackage{hyperref}

\begin{document}

\title{Adaptive arrival cost update for improving Moving Horizon Estimation performance}

\author[1]{Guido Sanchez}
\author[1]{Marina Murillo}
\author[1]{Leonardo Giovanini}

\affil[1]{Research Institute for Signals, Systems and Computational Intelligence, sinc(i), FICH-UNL/CONICET, Ciudad Universitaria UNL, $4^\circ$ piso FICH, (S3000) Santa Fe, Argentina.}

\date{March 2017}

\maketitle

\begin{abstract}
Moving horizon estimation is an efficient technique to estimate states and parameters of constrained dynamical systems. It relies on the solution of a finite horizon optimization problem to compute the estimates, providing a natural framework to handle bounds and constraints on estimates, noises and parameters. However, the approximation of the arrival cost and its updating mechanism are an active research topic. The arrival cost is very important because it provides a mean to incorporate information from previous measurements to the current estimates and it is difficult to estimate its true value. In this work, we exploit the features of adaptive estimation methods to update the parameters of the arrival cost. We show that, having a better approximation of the arrival cost, the size of the optimization problem can be significantly reduced guaranteeing the stability and convergence of the estimates. These properties are illustrated through simulation studies.

\end{abstract}

%\begin{keyword}
%Moving Horizon Estimation, Arrival cost, Adaptive estimation, Constrained state estimation, Variable forgetting.
%\end{keyword}

\section{Introduction} \label{sec1}
In control engineering, model--based control schemes assume that the states and parameters of the system are available for control law implementation. In practice, noisy measurements is the only information available from the system. Thus, the states and parameters have to be determined from these measurements using a dynamic model of the system. For linear systems, this problem has been solved and several methods based on different statistical measures have been developed \cite{jazwinski1970stochastic, crassidis2011optimal}.

The physical limits of the system can be modeled through bounds on its states and parameters. The omission of such information in the estimation algorithm may substantially hamper its performance \cite{haseltine2005critical}. Unfortunately, Kalman filter can not handle explicitly bounds on estimates (states and parameters) and ad hoc methods have been developed to handle constraints \cite{simon2010kalman}. The different approaches to enforce constraints in Kalman filter include model reduction \cite{simon2002kalman}, estimate projection \cite{hall2004mathematical}, gain projection \cite{teixeira2008gain}, probability density function truncation \cite{simon2010constrained} and system projection \cite{ko2007state}. These methods yield to suboptimal solutions at the best and unrealistic estimates when the statistics of unknown variables (initial states, measurement and process noises, disturbances) are poorly chosen. On the other hand, moving horizon estimation (MHE) solves an optimization problem to find the system estimates, providing a theoretical framework for constrained estimation.

\begin{figure}
\centering
    \includegraphics[width=0.85\textwidth]{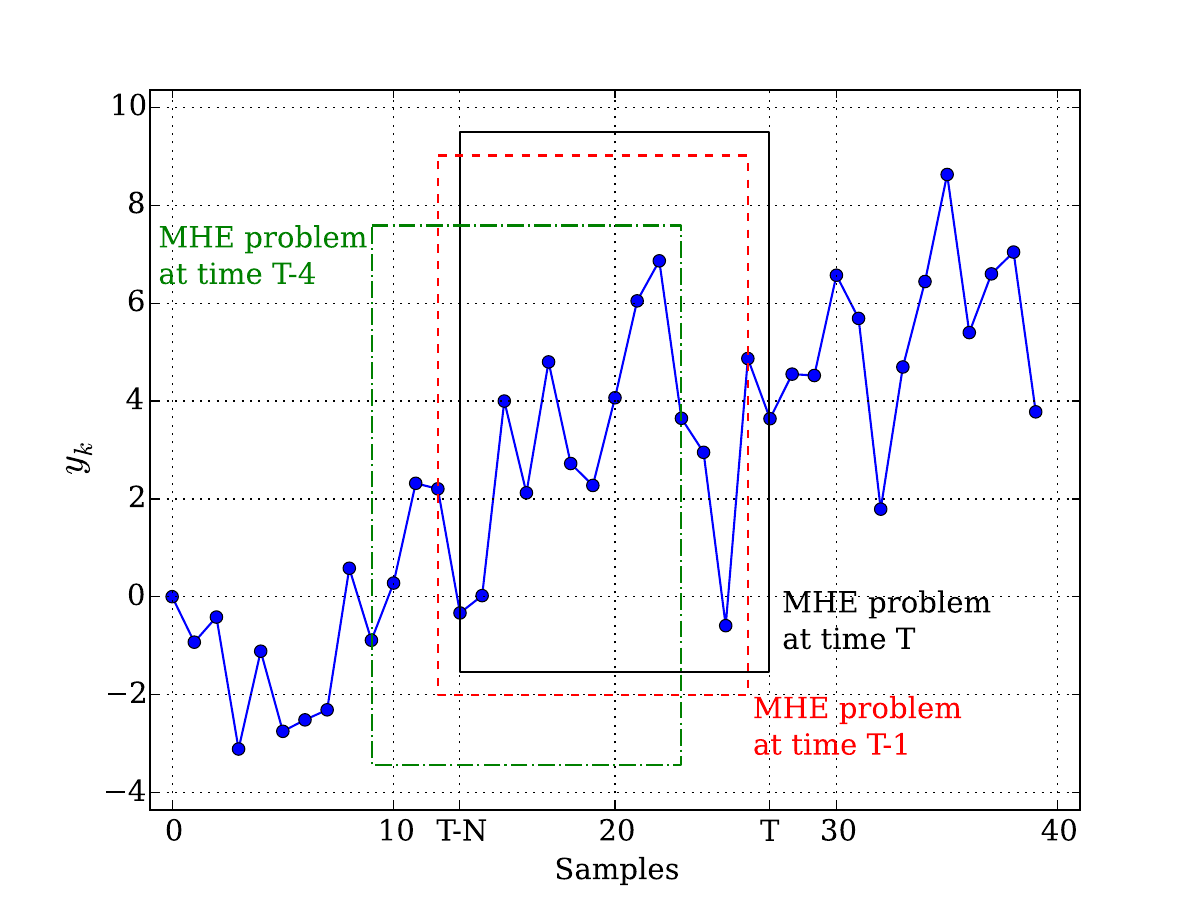}
    \caption{MHE smoothing update}
    \label{fig:smoothing}
\end{figure}

MHE solves at each sample a finite horizon state estimation problem to determine the states and parameters of the system.  When new measurements become available the old ones are discarded from the estimation window, and the estimation problem is solved to determine the new estimates (see Figure \ref{fig:smoothing}). The information of measurements which are not included in the estimation window is assimilated into the objective function through an extra term called \textit{arrival cost}. It characterizes the statistical distribution of states at the beginning of the estimation window given the prior measurements information. In this way, it allows MHE to approximate the \textit{full information} problem only considering a finite number of samples. A good approximation of the arrival cost allows to reduce the size of the estimation window and the size of the optimization problem, while it retains a good performance and robustness. The most accepted way of approximating the arrival cost is using a weighted 2-norm of the states at the beginning of the estimation window \cite{rao2001constrained, qu2009computation, kuhl2011real, chu2012moving}. For linear systems Rao et al. \cite{rao2001constrained} proposed to update the parameters of the arrival cost term (the weight matrix and initial states) using a Kalman filter or a Kalman smoother. The Gaussian distribution employed by these methods to model the conditional probabilities densities of initial states \cite{rao2001constrained, rao2002constrained} results in inadequate approximation of the arrival cost when estimates are constrained, leading to poor and unrealistic estimates. This problem arises from the fact that the presence of bounds on estimates modifies the probability distribution of noises and estimates, forcing to zero the probability of some values and eliminating the independence between estimates and noises \cite{robertson2002use}.

To tackle these problems, Chu et al. \cite{chu2012moving} derived an iterative arrival cost update scheme that uses a quadratic approximation and information of active/inactive constraints from previous iteration.  These ideas allow to build a quadratic approximation in the proximity of the optimal solution of the exact arrival cost. This update scheme is based on the hypothesis that the set of active constraints does not change after a particular time as the estimation horizon grows. This hypothesis works well when the estimation window is large, giving good results. However, if the hypothesis is not verified, i.e. some constraints become inactive after the state smoothing, the estimator may diverge. Furthermore, this updating mechanism overweights past data (by retaining active constraints) de-emphasizing the effect of information available in the new data on estimates. This fact can also cause divergence of the estimator if estimates are strongly correlated in time.

Finally, in a recent work Al-Matouq and Vincent \cite{al2015multiple} developed a MHE algorithm based on multiple estimation windows. The algorithm takes advantage of constraint inactivity to reduce the size of the optimization problem while it retains the stability and performance properties of full--information estimator. These properties are preserved by approximating the arrival cost with an unconstrained estimator that reformulates at each sample the objective function in the regions of constraints inactivity, allowing efficient long estimation windows.

In this work, the weighting matrix is updated using adaptive estimation algorithms in combination with the MHE filter solution in the previous sample, while the initial states are the MHE filter solution in the previous sample. It should be noted that approximating the arrival in this way is consistent with constraints and bound on estimates and noises, while at the same time it introduces a feedback mechanism between measured data and estimates, improving the overall performance of the estimator. In this way, the weighting matrix is computed in a closed-loop fashion rather than in an open-loop way, as it is done in standard estimation techniques. The efficiency of the proposed approach is evaluated by conducting simulations studies on a benchmark problem available in the literature. The main  contributions of  this paper  are to  show that  \textit{i)}  it is  important to  take into account bounds on estimates when the arrival cost is updated, \textit{ii)} this idea can be implemented using  adaptive  estimation  techniques  to  construct  an  approximation  of  the  arrival  cost,  \textit{iii)}  the estimation resulting from a MHE with the proposed updating scheme is stable, and \textit{iv)} the proposed updating scheme allows to shorten the estimation window without sacrificing performance, reducing the size of the optimization problem. The paper is organized as follows: In Section \ref{sec2} the MHE problem is presented. In Section \ref{sec3} two different ways of updating the arrival cost are analyzed. Simulation results are discussed in Section \ref{sec4} and the conclusions are outlined in Section \ref{sec_concl}.

\noindent \textbf{Notation:} In the sequel $x \in \mathbb{R}^{n}$ is a column vector and its transpose is denoted as $x^T$, $\mathbf{x}_k$ denotes a sequence of vectors over a given index up to $k$, for example $\mathbf{x}_k = \{ x_j : j = 0, 1, \ldots, k \}$,
$A \in \mathbb{R}^{n \times m}$ is a $n \times m$ matrix and its transpose is denoted by $A^T$.
If $A \in \mathbb{R}^{n \times n}$ its inverse is denoted by $A^{-1}$.
Given a symmetric real matrix $P \in \mathbb{R}^{n \times n}$, it is said to be positive definite if for all $z \in \mathbb{R}^{n}$, $z^T A z > 0$. In the following we will use the short notation $P>0$ to denote that matrix $P$ is positive definite.
By $|| \cdot ||$ we denote the Euclidean vector or induced matrix norm. For $x \in \mathbb{R}^n$ and $P > 0 \in \mathbb{R}^{n \times n}$, we let $||x||_P = \sqrt{x^T P x}$.

\section{Problem formulation} \label{sec2}
Consider the linear time invariant discrete system
\begin{equation}
\label{eq:model}
\begin{split}
x_{k+1} &= A x_k + G w_k, \\
y_k &= C x_k + v_k,
\end{split}
\end{equation}
where $x_k \in \mathcal{X} \subseteq \mathbb{R}^{n_{x}}$ is the state vector, $w_k \in \mathcal{W} \subseteq \mathbb{R}^{n_{w}}$ is the state noise vector, $y_k \in \mathbb{R}^{n_{y}}$ is the measurement vector, $v_k \in \mathcal{V} \subseteq \mathbb{R}^{n_{y}}$ is the measurement noise vector and $k \in \mathbb{I}$. In the following, $\hat{w}_k$ and $\hat{v}_k$ are considered zero mean stationary stochastic disturbances with finite moments and the restriction sets $\mathcal{X}$, $\mathcal{W}$ and $\mathcal{V}$ are considered closed with $\mathbf{0}\in\mathcal{X}$, $\mathbf{0}\in\mathcal{W}$ and $\mathbf{0}\in\mathcal{V}$.

The full information estimator\footnote{Other authors use the term batch least squares} (\textit{FIE}) problem uses a sequence of measurements $y_k$ to find the state estimates $\hat{x}_k$ that solve the following optimization problem
\arraycolsep=1.4pt
\begin{equation}
\label{fie_01}
\begin{array}{c}
\underset{\hat{x}_{0|k}, \mathbf{\hat{w}}_{k}}
{\operatorname{min}} \Phi_k = \norm{ \hat{x}_{0|k} - \bar{x}_0 } _ {P_{0}^{-1}}^2 + \sum_{j=0}^{k} \norm{ \hat{w}_{j|k}}_{Q^{-1}}^2 + \norm{\hat{v}_{j|k}}_{R^{-1}}^2 \\[0.5cm]
\text{s.t.} \left\{
\begin{array}{l}
\begin{array}{rl} \hat{x}_{j+1|k} =& A \hat{x}_{j|k} + \hat{w}_{j|k}, \\
y_{j} =& C \hat{x}_{j|k} + \hat{v}_{j|k},
\end{array} \\
\hat{x}_{j|k} \in \mathcal{X}, \  \hat{w}_{j|k} \in \mathcal{W}, \  \hat{v}_{j|k} \in \mathcal{V},
\end{array} \right.
\end{array}
\end{equation}
where $Q^{-1}$ and $R^{-1}$ are symmetric positive definite matrices that penalize the estimated noise vectors $\hat{w}_{j|k}$ and the output prediction $\hat{v}_{j|k}$ error, respectively, and the pair $( \bar{x}_0,P_0^{-1} )$ summarizes the prior information at time $k=0$ where $\bar{x}_0$ is the current knowledge of the initial estimate and $P_0^{-1}$ is a symmetric positive definite weighting matrix. The solution of problem (\ref{fie_01}) yields smoothed estimates $\{ \hat{x}_{j|k}: j=0,1,\ldots,k-1 \}$, the filtered estimate $\hat{x}_{k|k}$ and sequence of estimated noise vectors, denoted as $\mathbf{\hat{w}}_{k}$ and $\mathbf{\hat{v}}_{k}$.

Since the full information estimator uses all measurements, as new ones become available the problem size grows with time making it intractable. MHE overcomes this problem by only considering a fixed amount of data and dynamic updates by sliding a window with time. The past measurements are taken into account through a penalty cost term $Z_{k-N}$ called ``\emph{arrival cost}'' \cite{rao2001constrained}. Then, the MHE problem consists in finding the states, state and measurement noises that solve the following optimization problem
\begin{equation}
\label{eq:mhe1}
\begin{array}{c}
\underset{\hat{x}_{k-N|k}, \mathbf{\hat{w}}_{k} }
{\operatorname{min}} \Psi_k^N = Z_{k-N}(\hat{x}_{k-N|k}) + \sum_{j=k-N}^{k} \norm{\hat{w}_{j|k}}_{Q^{-1}}^2 + \norm{\hat{v}_{j|k}}_{R^{-1}}^2 \\[0.5cm]
\text{s.t.} \left\{
\begin{array}{l}
\begin{array}{rl} \hat{x}_{j+1|k} =& A \hat{x}_{j|k} + \hat{w}_{j|k}, \\
y_{j} =& C \hat{x}_{j|k} + \hat{v}_{j|k},
\end{array} \\
\hat{x}_{j|k} \in \mathcal{X}, \  \hat{w}_{j|k} \in \mathcal{W}, \  \hat{v}_{j|k} \in \mathcal{V}.
\end{array} \right.
\end{array}
\end{equation}
The minimization of the full information problem (\ref{fie_01}) is equivalent to minimizing the fixed horizon problem (\ref{eq:mhe1}) under the condition that $Z_{k-N}$ is the exact arrival cost $Z^*_{k-N}$
\begin{equation}
\label{exact_arrival_cost}
\begin{array}{c}
Z^*_{k-N} =
\underset{\hat{x}_{0|k-N}, \mathbf{\hat{w}}_{k-N} }
{\operatorname{min}} \; \norm{\hat{x}_{0|k-N} - \bar{x}_0}_{P_{0}^{-1}}^2 + \sum_{j=0}^{k-N-1} \norm{\hat{w}_{j|k-N}}_{Q^{-1}}^2 + \norm{\hat{v}_{j|k-N}}_{R^{-1}}^2 \\[0.5cm]
\text{s.t.} \left\{
\begin{array}{l}
\begin{array}{rl} \hat{x}_{j+1|k-N} =& A \hat{x}_{j|k-N} + \hat{w}_{j|k-N}, \\
y_{j} =& C \hat{x}_{j|k-N} + \hat{v}_{j|k-N}, \end{array} \\
\hat{x}_{j|k-N} \in \mathcal{X}, \  \hat{w}_{j|k-N} \in \mathcal{W}, \  \hat{v}_{j|k-N} \in \mathcal{V}.
\end{array} \right.
\end{array}
\end{equation}
The exact arrival cost $Z^*_{k-N}$ is a complicated piecewise quadratic function, whose transitions between segments are fixed upon the state of the constraints (active/inactive) in problem (\ref{fie_01}). Therefore, computing $Z^*_{k-N}$ involves the solution of a constrained quadratic problem (QP) which size grows with the horizon $k-N$. Consequently, $Z^*_{k-N}$ can not longer be computed solving the optimization problem (\ref{exact_arrival_cost}) and an approximation $\tilde{Z}_{k-N}$ is used. In conventional MHE formulations, a long estimation window is generally used to reduce the effect of the errors introduced by an improper approximation of the arrival cost \cite{patwardhan2012nonlinear}. In this way, MHE uses all the information available for the system, discarding the one that has no effect in the estimation due to the fading effect\footnote{With fading effect we mean that if the system can be represented with a FIR model of order $N$, then an horizon of size $N$ is employed, thus the system is fully described inside the horizon.} (finite memory) of the system. On the other hand, if the estimation window is shortened, a better approximation of the arrival cost is needed such that it incorporates this information into the problem. For linear systems, it is generally assumed that the distribution is Gaussian and the weighting matrix in the arrival cost is updated using a Kalman filter \cite{muske1993receding}. Moreover, for unconstrained linear systems under the Gaussian assumption MHE is equivalent to the Kalman filter. The penalty matrix is generally calculated using dynamic programming techniques, and it turns out to be the covariance of the state estimate at time $k-N$ in the stochastic interpretation \cite{rao2001constrained}.

\section{Arrival cost update} \label{sec3}
The MHE update relies on representing the full information problem (\ref{fie_01}) in the form (\ref{eq:mhe1}), where $N$ is fixed. Then, for each $k>N$, the arrival cost $Z^*_{k-N}$ is computed using (\ref{exact_arrival_cost}) and a new solution for the MHE problem is obtained. To avoid the curse of dimensionality of problem (\ref{exact_arrival_cost}), it can be updated recursively by extending the partial minimization of (\ref{exact_arrival_cost}) to the next decision variable
\begin{equation}
\label{bellman1}
\begin{array}{c}
Z^*_{k-N} = \underset{\mathbf{\hat{w}}_k }
{\operatorname{min}} \; Z^*_{k-N-1} + \norm{ \hat{w}_{k-N-1|k}}_{Q^{-1}}^2  +  \norm{\hat{v}_{k-N|k}}_{R^{-1}}^2 \\[0.5cm]
\text{s.t.} \left\{
\begin{array}{l}
\begin{array}{rl} \hat{x}_{k-N|k} =& A \hat{x}_{k-N-1|k-1} + \hat{w}_{k-N-1|k}, \\
y_{k-N} =& C \hat{x}_{k-N|k} + \hat{v}_{k-N|k}, \end{array} \\
\hat{x}_{k-N|k} \in \mathcal{X}, \  \hat{w}_{k-N-1|k} \in \mathcal{W}, \  \hat{v}_{k-N|k} \in \mathcal{V}.
\end{array} \right.
\end{array}
\end{equation}
Starting with
\begin{equation}
Z^*_0 = \norm{ \hat{x}_0 - \bar{x}_0}_{P_0^{-1}}^2,
\end{equation}
since it is hard to find an analytic expression for the exact arrival cost of the constrained MHE problem, we will follow Rao et al. \cite{rao2001constrained} approach and use an approximate arrival cost based on the arrival cost for the unconstrained problem. Therefore, the exact arrival cost will be approximated by the following quadratic approximation
\begin{equation}
\label{z_recursion}
\tilde{Z}_{k-N} = \norm{ \hat{x}_{k-N} - \bar{x}_{k-N} }_{P_{k-N}^{-1}}^2,
\end{equation}
where $P_{k-N}$ and $\bar{x}_{k-N}$ fully define the approximate arrival cost. The update of the approximate arrival cost $\tilde{Z}_{k-N}$ can be formulated as an update of $P_{k-N}$ and $\bar{x}_{k-N}$. Considering the partial minimization problem for the approximate iteration (\ref{bellman1}), where $\tilde{Z}_{k-N-1}$ is given by (\ref{z_recursion}), the result is a constrained sequential estimation problem.

For unconstrained problems, the solution of (\ref{bellman1}) using (\ref{z_recursion}) leads to an unconstrained sequential estimation problem that can be solved using a Kalman filter \cite{muske1993receding}. Under this approach, the prior information is transferred to the current estimated window by conditioning the estimates of time $k$ using the optimal estimate $\hat{x}^*_{k-N-1|k-N-1}$. For constrained problems this idea can still be employed, replacing the unconstrained Kalman filter by any of the ad-hoc methods described in \cite{simon2010kalman}. However, the constraints are only applied to estimates and the weighting matrices remain unaffected. The major argument in favor of this kind of update lays in the fact that no input information $y_j$ is over weighted during the estimation process, which guarantees that no measurement information is used twice. One disadvantage of using the filtering update lays in the fact that a periodic behavior associated with relying on estimates based on few data points occurs. This phenomenon is known as \emph{estimation cycling} and can be avoided using the smoothing update \cite{tenny2002efficient}. The smoothing formulation takes advantage of more information by including more data in the $\bar{x}_{k-N}$ update. Another problem is the inaccurate approximation of the arrival cost when inequality constraints switch between being active and inactive at the optimal solution. Constraints modify the probability distribution and noises in a way that analytical updates can not properly approximate. In this work, the approximate arrival cost update (\ref{bellman1}) is formulated as an update of $P_{k-N}$ and $\bar{x}_{k-N}$ in (\ref{z_recursion}) in two steps: \textit{i}) update $\bar{x}_{k-N}$ and then \textit{ii}) update $P_{k-N}$.

\subsection{Computing $\bar{x}_{k-N}$}

One of the components of $\tilde{Z}_{k-N}$ is the initial state $\bar{x}_{k-N}$. Since the MHE problem yields smoothed estimates $\{ \hat{x}_{j|k}: j=k-N,\ldots,k-1 \}$ and the filtered estimate $\hat{x}_{k|k}$, we can use the one time-step before optimal smoothed state estimate $\hat{x}^*_{k-N|k-1}$,
\begin{equation*}
\bar{x}_{k-N} = \hat{x}^*_{k-N|k-1}.
\end{equation*}
The use of the smoothed update implies the use of similar approximation hypothesis like Chu et al. \cite{chu2012moving}: the state of constraints (active/inactive) does not change once the estimate leaves the estimation window.

\subsection{Computing $P_{k-N}$}
The weighting matrix $P_{k-N}$ can be seen as the covariance of $x_{k-N}$ in the stochastic interpretation of the MHE problem \cite{rao2001constrained}. Analytical solutions for computing $P_{k-N}$ only exists for linear unconstrained problems. For linear constrained problems (and nonlinear problems) approximate solutions have been derived. They are based on a recursive update that employs the information of the system available at time $k$ (system model, previous covariance matrix, active constraints, covariance of the noises). These equations evolve in an open-loop fashion, based on the system model and the assumed statistics of the noises, without incorporating any information from measurements or estimates.

In signal processing and adaptive estimation a recursive update based on measurements and/or estimates is often used \cite{landau1998adaptive}
\begin{equation}  \label{cov_01}
   P^{-1}_{k-N} = \alpha_k P^{-1}_{k-N-1} + \beta_k \hat{x}_{k-N|k-1} \hat{x}^T_{k-N|k-1} \quad P_{0} > 0,
\end{equation}
where $0 < \alpha_{k} < 1 \text{ and } 0 < \beta_{k} \leq 2 \; \forall k\geq 0$ are time-varying weighting sequences. Equation (\ref{cov_01}) generates recursively a real-time estimation of the weighting matrix $P_{k-N}$ by updating the previous estimates $P_{k-N-1}$ with an exponential time averaging of $\hat{x}_{k-N|k-1} \hat{x}^T_{k-N|k-1}$. This updating mechanism can be understood as a time-varying filter whose input is $\hat{x}_{k-N|k-1} \hat{x}^T_{k-N|k-1}$ and the initial condition is $P_{0}$.

% \begin{lemma}
% \label{lemma_outer_prod}
% Let $x \in \mathbb{R}^n$, then $x x^T$ is positive semidefinite.
% \end{lemma}

% \begin{proof}
% The matrix $X=x x^T$ is symmetric, since
% \begin{equation*}
% X^T = (x x^T)^T = (x^T)^T x^T = x x^T = X
% \end{equation*}
% Given that $X$ is symmetric, we will show that it is positive semidefinite. Let $z \in \mathbb{R}^n$ be nonzero, then
% \begin{equation*}
%     z^T X z = z^T x x^T z = (x^T z)^T (x^T z) = ||x^T z||^2_2 \geq 0
% \end{equation*}
% \end{proof}

\begin{lemma}
\label{ob_sum_pos}
Let $\{ A_i : i=1,\ldots, n\}$ be positive definite and let $\{ \gamma_i : i=1,\ldots, n\}$ be nonnegative real numbers. Then $\sum_{i=1}^{n} \gamma_i A_i$ is positive semidefinite.
\end{lemma}

\begin{proof}
Let $x \in \mathbb{R}^n$ be nonzero and observe that
$x^T(\sum_{i=1}^{n} \gamma_i A_i)x = \sum_{i=1}^{n} \gamma_i(x^T A_i x) \geq 0$. Since each $\gamma_i \geq 0$ and each $x^T A x > 0$. The latter sum results in a positive definite matrix if any of the summands is positive.
\end{proof}

\begin{lemma} If $P_{0} > 0$ and $\alpha_0 > 0$, then the matrix $P_{k}$ obtained using Eq. (\ref{cov_01}) is positive definite for all $k \geq 0$.
\label{lem:p_defpos}
\end{lemma}

\begin{proof}
Since $x x^T$ is positive semidefinite and from Lemma \ref{ob_sum_pos} it follows that $P_k$ is positive definite for all $k \geq 0$.
\end{proof}

\begin{lemma}
\label{lemma_inverse}
If $P>0$, then there exists an invertible matrix $Z$ such that
$Z^T P Z = I$ and $Z Z^T = P^{-1}$.
\end{lemma}
\begin{proof}
See Lemma 7, page 261 from Johnson \cite{Johnson1970}.
\end{proof}

Note that $\alpha_k$ and $\beta_k$ in (\ref{cov_01}) have opposite effects. When $\alpha_k \ll 1$ the adaptation algorithm tends to discard the old information summarized in $P_{k-N-1}$, estimating $P_{k-N}$ using the recent data $\hat{x}_{k-N|k-1}$ . This fact allows to avoid the overweighting of past data (by retaining data with active constraints) that can deteriorate the estimator performance, improving the estimation of $\hat{x}_{k-N|k}$. From an optimization point of view, when the algorithm discards the data without useful information, the matrix $P^{-1}_{k-N}$ increases and the initial information available in the arrival cost tends to be more relevant. On the contrary, when $\alpha_k \approx 1$ the adaptation algorithm tends to retain the old data summarized in $P_{k-N-1}$.  This fact allows to include all the past information that is relevant for the estimation problem, improving  the estimation of $\hat{x}_{k-N|k}$ by including more information and forgetting only initial data, facts that decrease $P^{-1}_{k-N}$. When $\beta_k \gg 0$ the adaptation algorithm tends to increase $P^{-1}_{k-N}$ by incorporating information from estimates $\hat{x}_{k-N|k-1}$. This fact allows to include more information from recent data into $P^{-1}_{k-N}$, increasing its value. When $\beta_k \approx 0$ the adaptation algorithm discards the information available in $\hat{x}_{k-N|k-1}$ and only applies a forgetting mechanism on $P^{-1}_{k-N}$, which depends on the value of $\alpha_k$. In general, $\alpha_k$ and $\beta_k$ allow the estimation algorithm to discard and/or incorporate information into $P^{-1}_{k-N}$ according to the information available in the measurements $\{ y_j \}_{j=k-N}^{k}$, estimates $\hat{x}_{k-N|k}$ and process noises $\{ w_j \}_{j=k-N}^{k}$.

Since the existence of the inverse is guaranteed by Lemma \ref{lemma_inverse}, we can use the matrix inversion lemma to rewrite Eq. (\ref{cov_01}) as follows
\begin{equation}  \label{cov_02a}
  P_{k-N} = \frac{1}{\alpha_k} \left[ P_{k-N-1} - \frac{P_{k-N-1} \hat{x}_{k-N|k-1} \hat{x}^T_{k-N|k-1} P_{k-N-1}}{ \frac{\alpha_k}{\beta_k} + \hat{x}^T_{k-N|k-1} P_{k-N-1} \hat{x}_{k-N|k-1}} \right].
\end{equation}
This way of computing $P_{k-N}$ introduces a feedback loop into the MHE problem. This fact allows the MHE problem to adjust $P^{-1}_{k-N}$ according to the information available in the measurements $\{ y_j \}_{j=k-N}^{k}$ and to guarantee the convergence of estimates $\hat{x}_{k-N|k}$ and $\{ \hat{w}_j \}_{j=k-N}^{k}$. We must be careful in the choice of the arrival cost update mechanism, since poor approximations can lead to instability \cite{findeisen1997moving, rao2000moving} in the estimator. In the following Sections the conditions that sequences $\alpha_k$ and $\beta_k$ must satisfy to guarantee convergence will be studied.
If the current sequence of process noise $\{ w_j \}_{j=k-N}^{k}$ is large, $P^{-1}_{k-N}$ will grow emphasizing $\tilde{Z}_{k-N}$ over the other terms. This fact will lead to $\tilde{Z}_{k-N} \to 0$ and improving the overall estimation. Then, $P^{-1}_{k-N}$ will decrease until $\tilde{Z}_{k-N}$ achieves its stationary value or the feedback starts again (see Fig \ref{fig:pinv_varvar}).

Depending on the values of $\alpha_k$ and $\beta_k$, different forgetting factor profiles, suitable for different adaptation problems, are obtained \cite{landau1998adaptive}. For example, if the estimates are varying slowly and constraints does not change their state after leaving the estimation window, the most appropriate selection is
\begin{equation}  \label{f_ff}
     0<\alpha_k=\alpha<1, \quad \beta_k=1,
\end{equation}
which leads to a \textit{constant forgetting factor}. In this case, equation (\ref{cov_02a}) will estimate $P_{k-N}$ using the available MHE estimates, de-emphasizing past observations at a constant rate.
When a constant forgetting factor is used, the choice of $\alpha_k$ and $\beta_k$ given by Eq (\ref{f_ff}) leads to a decreasing adaptation gain, where old information is continually forgotten. If the forgetting factor $\alpha$ is not carefully chosen in this schema, this may lead to an exponential growth of the matrix $P_{k-N}$ and it can turn the MHE algorithm extremely sensitive to disturbances and susceptible to numerical difficulties \cite{fortescue1981implementation}.

This is why several researchers \cite{aastrom1973self, osorio1981deterministic, Ljung1990} suggest to use a variable forgetting factor, which enables the algorithm to follow slow and sudden changes in the system, while it can prevent matrix $P_{k-N}$ from blowing-up.
This choice is appropriate when the estimates are stationary  and constraints can change their state after leaving the estimation window \cite{landau1998adaptive}, since some forgetting must be done in order to maintain the information of the relevant constrained states.
In this work we review some of the many mechanisms available to update the parameters $\alpha_k$ and $\beta_k$.
For example, if the evolution of $\alpha_k$ is given by
\begin{equation}  \label{v_ff}
  \alpha_k = 1 - \frac{ \hat{x}^T_{k-N|k-1}P_{k-N-1} \hat{x}_{k-N|k-1}}{1 + \hat{x}^T_{k-N|k-1} P_{k-N-1} \hat{x}_{k-N|k-1} }, \quad \beta_k = 1.
\end{equation}
We can see that the forgetting factor in Eq (\ref{v_ff}) depends on estimates, leading to a \textit{variable forgetting factor}. It automatically takes the value $1$ if the norm of the states becomes null, leading to the standard Recursive Least Squares (RLS) algorithm. In the cases where the states sequence is such that the term $\hat{x}^T_{k-N|k-1} P_{k-N-1} \hat{x}_{k-N|k-1}$ is significative with respect to one, the forgetting factor takes a lower value assuring good adaptation capabilities.
The algorithm will provide a quick adaptation and will prevent the exponential growth of $P_{k-N}$, stabilizing the estimates (controlling the amount of information employed to compute $P_{k-N}$) and the MHE problem (\ref{eq:mhe1}).

If the estimates are time--varying and constraints can change after leaving the estimation window, the forgetting factor has to guarantee a constant trace of the weighting matrix, which can be obtained by
\begin{equation}  \label{ct_ff}
  \alpha_k = \frac{1}{\Xi} \,
  \Tr{\left( P_{k-N-1} - \frac{ P_{k-N-1} \hat{x}_{k-N|k-1} \hat{x}^T_{k-N|k-1} P_{k-N-1} }{ \eta + \hat{x}^T_{k-N|k-1} P_{k-N} \hat{x}_{k-N|k-1} } \right)}, \quad \beta_k = \eta^{-1} \alpha_k,
\end{equation}
where $\Xi > 0$ is the desired trace of the weighting matrix and $\eta > 0$ is constant.

Fortescue et al. \cite{fortescue1981implementation} proposed a variable forgetting algorithm that tries to keep the amount of information constant and shows that a reasonable choice of information measure can prevent the weighting matrix from blowing-up, while retaining the adaptability. Later, Osorio Cordero and Mayne \cite{osorio1981deterministic} modified that algorithm in order to ensure global convergence. The modifications ensure that $\Tr{P_k}$ (and, hence $||P_k||$) is always less than a given constant $c$. The algorithm yields the standard least-squares estimate if $\alpha_k = 1$, and the weighted least-squares estimate if $\alpha_k \in (0,1)$.
The memory length $N_k$ is related to the forgetting factor $\alpha_k$ through $\alpha_k = 1 - 1/N_{k}$; forgetting factors of $0$, $0.99$ and $1$ correspond to memory lengths of $1$, $100$ and $\infty$, respectively.
The various constant required by the algorithm may be chosen as follows: $\sigma$ is chosen so that $\sigma/\sigma_w$ is large, where $\sigma_w$ denotes the process noise variance; $P_0 = \delta I$, where $\delta$ is large ($10^6$, for example), $c>\delta$ is large. The variable forgetting factor and the information matrix can be calculated as follows
\begin{align}
\begin{split}
\label{vf_mayne}
\epsilon_{k} &= y_{k-N} - \hat{y}_{k-N} \\
% K_k &= [1 + \hat{x}_{k-N|k-1}^T P_{k-N-1} \hat{x}_{k-N|k-1} ]^{-1} P_{k-N-1} \hat{x}_{k-N|k-1} \\
N_{k} &= \left. [1+\hat{x}_{k-N|k-1}^{T} P_{k-N-1} \hat{x}_{k-N|k-1} ] \frac{\sigma}{||\epsilon_{k-N}||^2_2} \right. \\
\alpha_k &= 1 - 1/N_{k} \\
W_{k} &= [I - \frac{P_{k-N-1} \hat{x}_{k-N|k-1} \hat{x}_{k-N|k-1}^T}{1 + \hat{x}_{k-N|k-1}^T P_{k-N-1} \hat{x}_{k-N|k-1}}] P_{k-N-1} \\
& P_{k-N} = \frac{1}{\alpha_k} W_k, \text{ if } 1/\alpha_k \Tr(W_k) \leq c \\
& P_{k-N} = W_k, \text{ otherwise}.
\end{split}
\end{align}

\begin{remark}
The weighting matrix estimation algorithm can be recasted like a Kalman filter with the smoothed estimate $\hat{x}_{k-N|k-1}$ as estimated states, $\hat{v}_k$ as the estimation residual and $P_{k-N}$ the covariance matrix \cite{landau1998adaptive}.
In case of unconstrained estimation, the weighting matrix estimation algorithm is equivalent to a filter update  \cite{muske1993receding}. The relationship between RLS and the Kalman filter is well known \cite{jazwinski1970stochastic}, where in the case of the Kalman filter, the weighting matrix update formula is a special case of Eq. ($\ref{cov_01}$) under specific assumptions and a certain choice of $\alpha_k$ and $\beta_k$. For a detailed description the reader can see Section 3.2.4 of Landau et al. \cite{landau1998adaptive}.
\end{remark}

\subsection{MHE update summary}

The MHE approach solves problem (\ref{eq:mhe1}), providing a solution $\{\hat{x}_{j|k}\}_{j=k-N}^k$  that approximates the optimal solution of problem (\ref{fie_01}). For $k \leq N$, the full information problem (\ref{fie_01}) is solved. Then, for $k>N$ the MHE problem (\ref{exact_arrival_cost}) is solved with the approximate arrival cost
\begin{equation}
\tilde{Z}_{k-N} = \norm{ \hat{x}_{k-N|k} - \hat{x}_{k-N|k-1}^* }_{P_{k-N}^{-1}}^2,
\end{equation}
whose solution can be obtained using any suitable QP solver. Then, $P_{k-N}$ and $\bar{x}_{k-N}$ are updated and the horizon is receded, updating the measurements vector $\{y_j\}_{j=k-N}^k$. Algorithm \ref{algo_mhead} summarizes the proposed moving horizon estimator. The method used for the adaptive weighting matrix update defines the algorithm behavior. If we wish to use a constant trace approach ($\mathrm{MHE_{AD-CT}}$), Eq. (\ref{ct_ff}) is used. On the other hand, if we wish to use a variable forgetting factor ($\mathrm{MHE_{AD-VF}}$), Eq. (\ref{vf_mayne}) is used.

\begin{algorithm}
\caption{$\mathrm{MHE_{AD}}$ Algorithm}
\label{algo_mhead}
\begin{algorithmic}
\STATE $\bar{x}_{0}$, $P_0>0$, $Q,R>0$ and $N\geq1$
\STATE \textbf{Initialization:}
\STATE $\bar{x}_{k-N} \gets \bar{x}_0$
\STATE $P_{k-N} \gets P_0$
\FOR{samples $k=0,1,2,\ldots$}
\STATE Obtain measurement $y_k$
\IF{$k < N$}
	\STATE Prepare data vector:
	\STATE $\qquad Y \gets [y_0,\ldots,y_k]^T$
	\STATE Solve $\Phi_k$ with $\bar{x}_{0}$, $P_0$ and $Y$
\ELSE
	\STATE Prepare data vector:
	\STATE $\qquad Y \gets [y_{k-N},\ldots,y_k]^T$
	\STATE Solve $\Psi^N_k$ with $\bar{x}_{k-N}$, $P_{k-N}$ and $Y$
	\STATE $\bar{x}_{k-N} \gets \hat{x}_{k-N+1|k}$
	\STATE Update $P_{k-N}$ using Eq. (\ref{ct_ff}) or (\ref{vf_mayne})
\ENDIF
	\STATE Obtain current estimates $\hat{x}_{j|k}$ for $j=k-N,\ldots,k$
\ENDFOR
\end{algorithmic}
\end{algorithm}

\section{Stability analysis}

Stability of the estimator implies that the reconstruction error converges to zero for the nominal system
\begin{align}
\begin{split}
x_{k+1} &= A x_k, \\
y_k &= C x_k,
\end{split}
\label{eq:nominal_system}
\end{align}
with no state or measurement noise. However, for the constrained case, a poor choice of constraints may prevent convergence to the true state of the system. Rao et al. \cite{rao2001constrained} added the requirement that the evolution of the system respects the constraints, which need to satisfy only the following weaker assumption to prove stability.

\begin{assumption}
\label{assumption_existence}
Suppose the system (\ref{eq:nominal_system}) with initial condition $x_0$ generates the data. We assume there exists $x_{0|\infty}$, $\{w_{k|\infty}\}_{k=0}^{\infty}$ and $\sigma > 0$ such that
\begin{equation*}
\sum_{k=0}^{\infty} \norm{v_{k|\infty}}_{R^{-1}} + \norm{w_{k|\infty}}_{Q^{-1}} + \norm{x_{0|\infty} - \hat{x}_0}_{P_0^{-1}} \leq \sigma \norm{x_0 - \hat{x}_0}^2
\end{equation*}
and
$x_{k|\infty} \in \mathcal{X}$, $w_{k|\infty} \in \mathcal{W}$, $v_{k|\infty} \in \mathcal{V}$.
\end{assumption}

Assumption \ref{assumption_existence} states the existence of a feasible state and disturbance trajectory that yields bounded cost. Next we state the stability results for the FIE.

\begin{remark}
Suppose the matrices $Q$, $R$, and $P_0$ are positive definite, $(C,A)$ is observable, and assumption \ref{assumption_existence} holds. Then, the constrained full information estimator is an asymptotically stable observer for the system given by Eq. (\ref{eq:nominal_system}).
\end{remark}

\begin{proof}
See Muske et al. \cite{muske1993receding}.
\end{proof}

\begin{remark}
The matrix $P_k$ satisfies $||P_k|| \leq c, \, \forall k$. The boundedness of $P_k$ guarantees that the optimization problem for MHE with adaptive arrival cost update is convex and that under Assumption \ref{assumption_existence}, a solution exists.
\label{P_k_leq_c}
\end{remark}

\begin{proof}
From Lemma \ref{lem:p_defpos}, for all $k$, $P_k$ is symmetric and positive definite. Its norm $||P_k||$ is equal to $\lambda_{max}(P_k)$, the maximum eigenvalue of $P_k$. Its trace satisfies $||P_k||\leq \operatorname{trace}(P_k)$.

By choice $||P_0|| \leq \operatorname{trace}(P_0) < c$. Suppose $||P_k|| \leq c$. From Eq. (\ref{vf_mayne}), if $(1/ \alpha_k) \operatorname{trace} (W_k) \leq c$, then $||P_k||=(1/ \alpha_k)||W_k|| \leq (1/ \alpha_k)\operatorname{trace}(W_k)\leq c$. If, on the other hand, $(1/ \alpha_k) \operatorname{trace}(W_k)>c$, then $\alpha_k = 1$ and $P_k^{-1} = P_{k-1}^{-1} + \hat{x}_{k-N|k-1} \hat{x}_{k-N|k-1}^T$, i.e. $P_k$ is less positive definite than $P_{k-1}$, so that $||P_k|| \leq ||P_{k-1}|| \leq c$. Hence, $||P_k|| \leq c$.
\end{proof}

\begin{remark}
The matrix $P_k$ has a steady state-solution,
\begin{equation*}
    \underset{k \to \infty}{\lim} P_k = P_{\infty}
\end{equation*}
\end{remark}

\begin{proof}
The update formula for $P_{k-N}$ given by Eq. (\ref{vf_mayne}) can be seen as a special case of a Ricatti equation as follows
\begin{align*}
    \alpha_k P_{k-N} &= P_{k-N-1} - \frac{P_{k-N-1} \hat{x}_{k-N|k-1} \hat{x}^T_{k-N|k-1} P_{k-N-1}}{ 1 + \hat{x}^T_{k-N|k-1} P_{k-N-1} \hat{x}_{k-N|k-1}} \\
    &= P_{k-N-1} -
    P_{k-N-1} \hat{x}_{k-N|k-1}
    (1 + \hat{x}^T_{k-N|k-1} P_{k-N-1} \hat{x}_{k-N|k-1})^{-1}
    \hat{x}^T_{k-N|k-1} P_{k-N-1}
\end{align*}
which is stable if and only if all of the eigenvalues of the closed loop state transfer matrix
\begin{equation}
    I - \frac{\hat{x}_{k-N|k-1} \hat{x}_{k-N|k-1}^T P_{k-N-1}}{1 + \hat{x}^T_{k-N|k-1} P_{k-N-1} \hat{x}_{k-N|k-1}}
    \label{ricatti_transfer_matrix}
\end{equation}
are strictly inside the unit circle of the complex plane. The eigenvalues of the right hand side of the previous expression lie inside the unit circle, since
\begin{align*}
    || \hat{x}_{k-N|k-1} \hat{x}_{k-N|k-1}^T P_{k-N-1} || \leq ||\hat{x}_{k-N|k-1}|| \ ||\hat{x}_{k-N|k-1}^T|| \ ||P_{k-N-1}||
\end{align*}
and
\begin{align*}
    || \hat{x}^T_{k-N|k-1} P_{k-N-1} \hat{x}_{k-N|k-1} || \leq ||\hat{x}_{k-N|k-1}^T|| \ ||P_{k-N-1}|| \ ||\hat{x}_{k-N|k-1}||
\end{align*}
hence
\begin{align}
    \begin{split}
    \left\lVert  \frac{\hat{x}_{k-N|k-1} \hat{x}_{k-N|k-1}^T P_{k-N-1}}{1 + \hat{x}^T_{k-N|k-1} P_{k-N-1} \hat{x}_{k-N|k-1}} \right\rVert \leq 1.
    \end{split}
    \label{ricatti_transfer_matrix_rhs_norm}
\end{align}
Since the product of symmetric positive semidefinite matrices is another symmetric positive semidefinite matrix, the eigenvalues of (\ref{ricatti_transfer_matrix_rhs_norm}) are all greater than $0$ and bounded by $1$. Then, the eigenvalues of (\ref{ricatti_transfer_matrix}) satisfy that are greater than 0 and bounded by $1$.
Hence, the matrix $P_k$ has a steady-state solution $P_\infty$.
\end{proof}

\begin{theorem}
$V_k = \tilde{x}_{k}^{T} P_{k}^{-1} \tilde{x}_{k}$ is a Lyapunov function for MHE with adaptive arrival cost update, where $\tilde{x}_{k} = x_k - \hat{x}_k$ and the matrix $P_k^{-1}$ is given by Eq. (\ref{vf_mayne}), i.e. $P_{k}^{-1} = \alpha_k (P_{k-1}^{-1} + \hat{x}_{k-N|k-1} \hat{x}_{k-N|k-1}^T)$.
\end{theorem}

\begin{proof}
Since $\alpha_k V_{k-1} \leq V_{k-1}$, we can subtract $\alpha_k V_{k-1}$ to $V_{k}$
\begin{align*}
V_{k} - \alpha_k V_{k-1} &= \tilde{x}_{k}^{T} P_{k}^{-1} \tilde{x}_{k}  - \alpha_k \tilde{x}_{k-1}^{T} P_{k-1}^{-1} \tilde{x}_{k-1} \\
    &= (A \tilde{x}_{k-1})^{T} P_{k}^{-1} (A \tilde{x}_{k-1})  - \alpha_k \tilde{x}_{k-1}^{T} P_{k-1}^{-1} \tilde{x}_{k-1} \\
    &= \tilde{x}_{k-1}^{T} A^{T} P_{k}^{-1} A \tilde{x}_{k-1}  - \alpha_k \tilde{x}_{k-1}^{T} P_{k-1}^{-1} \tilde{x}_{k-1} \\
    &= \tilde{x}_{k-1}^{T} A^{T} [\alpha_k (P_{k-1}^{-1} + \hat{x}_{k-N|k-1} \hat{x}_{k-N|k-1}^T)] A \tilde{x}_{k-1}  - \alpha_k \tilde{x}_{k-1}^{T} P_{k-1}^{-1} \tilde{x}_{k-1} \\
    &= \alpha_k \tilde{x}_{k-1}^{T} [A^{T} P_{k-1}^{-1} A  + A^{T} \hat{x}_{k-N|k-1} \hat{x}_{k-N|k-1}^T A - P_{k-1}^{-1}] \tilde{x}_{k-1}
\end{align*}
Since $P_k$ is symmetric positive definite for all $k$ and $A^T \hat{x}_{k-N|k-1} \hat{x}_{k-N|k-1}^T A$ is symmetric positive semidefinite for all $k$, it follows that if the system given by Eq. (\ref{eq:nominal_system}) is observable, then MHE with adaptive arrival cost update is an asymptotically stable observer.
\end{proof}

\subsection{Bounded stability}

From Remark \ref{P_k_leq_c} we know that $||P_k||$ is nonincreasing, hence $||P_k^{-1}||$ is nondecreasing. By the convergence of $P_k$ to $P_\infty$, we get that $||P_k^{-1}|| \leq ||P_{k+1}^{-1}|| \leq ||P_{\infty}^{-1}||$. Therefore, the arrival and stage cost satisfy the following bounds
\begin{equation*}
    \lambda_{min}(P_{k-N}^{-1}) ||\hat{x}_{k-N|k} - \bar{x}_{k-N}|| \leq ||\hat{x}_{k-N|k} - \bar{x}_{k-N}||_{P_{k-N}^{-1}} \leq \lambda_{max}(P_{\infty}^{-1}) ||\hat{x}_{k-N|k} - \bar{x}_{k-N}||
\end{equation*}
\begin{equation*}
\lambda_{min}(Q^{-1}) ||\hat{w}|| \leq ||\hat{w}||_{Q^{-1}} \leq \lambda_{max}(Q^{-1}) ||\hat{w}||
\end{equation*}
\begin{equation*}
\lambda_{min}(R^{-1}) ||\hat{v}|| \leq ||\hat{v}||_{R^{-1}} \leq \lambda_{max}(R^{-1}) ||\hat{v}||.
\end{equation*}
Under bounded disturbances, the MHE estimator with adaptive arrival cost update is robust globally asymptotically stable. Furthermore, if $\operatorname{lim}_{k \to \infty} w_k = 0$ and $\operatorname{lim}_{k \to \infty} v_k = 0$, then $\operatorname{lim}_{k \to \infty} ||\hat{x}_k - \bar{x}_k || = 0$.
\begin{proof}
See Theorem 7 of \cite{Muller2016}.
\end{proof}

\section{Results}
\label{sec4}

For the simulations tests we adopted the discrete-time model used by Rao et al. in  \cite{rao2003constrained}
\begin{equation}
\label{results_model}
\begin{split}
x_{k+1} &= \left[ \begin{matrix} 0.99 & 0.2 \\ -0.1 & 0.3 \end{matrix} \right] x_k + \left[ \begin{matrix} 0 \\ 1 \end{matrix} \right] w_k, \\
y_k &= \left[ \begin{matrix} 1  & -3 \end{matrix} \right] x_k + v_k,
\end{split}
\end{equation}
where $\{v_k\}$ is a sequence of independent, normally distributed random variables with zero mean and covariance $\sigma_v$, and $w_k = |z_k|$ where $\{z_k\}$ is a sequence of independent, normally distributed random variables with zero mean and covariance $\sigma_w$. For all the following examples, the initial state $\bar{x}_0$ is normally distributed with zero mean and covariance $\sigma_x=0.5$ and we assume that $x_0 = [0,0]^T$ and $\bar{x}_0 = [0.5,-0.5]$. The constrained estimation problem is formulated here with $Q=1$, $R=0.01$, $P_0 = 0.5$ and we add the inequality constraint $w_k \geq 0$ to include the available information on the random sequence.

First, the proposed algorithms ($\mathrm{MHE_{AD-VF}}$ and $\mathrm{MHE_{AD-CT}}$)  are compared with the full information estimator (FIE) and with a linear algorithm ($\mathrm{MHE_{KF}}$) available in the Nonlinear Model Predictive Control Tools for CasADi (mpc-tools-casadi) toolbox \cite{mpc-tools-casadi}. This algorithm propagates the arrival cost and initial state estimate using a Kalman filter. The simulations evaluates the effect of the horizon size on the performance of the estimators. Three different horizon sizes ($N=3$, $N=6$ and $N=10$) were used for each MHE algorithm with $\sigma_v = 0.1$ and $\sigma_w=1.0$. The sum square estimation error
\begin{equation}
\xi^{(i)} = \sum_{j=0}^{k}(x^{(i)}_{j} - \hat{x}^{(i)}_{j})^2, \nonumber
\end{equation}
where the superscript $(i)$ denotes the $i$th component of $x$, was used as benchmark. In this test, 100 trials were run and the average sum square estimation error was computed. The results are shown in Tables \ref{tab:200trials_x1} and \ref{tab:200trials_x2}.
It can be seen that $\mathrm{MHE_{AD}}$ errors are always lower than $\mathrm{MHE_{KF}}$ and closer to FIE on every window size considered. As the window size is enlarged, the arrival cost plays a less significant role on the cost function of the optimization problem and the cost of the three estimators tends to converge to the same value. A good approximation of the arrival cost allows us to decrease the window size. In this aspect, it can be seen that $\mathrm{MHE_{AD}}$ errors do not change significantly with the different horizon sizes, showing that this algorithm seems to be a good method for the arrival cost approximation, even for short horizons.

\begin{table}
\begin{center}
\begin{tabular}{ccccc}
\hline
& \multicolumn{4}{c}{$x^{(1)}$} \\

& FIE & $\mathrm{MHE_{AD-VF}}$ & $\mathrm{MHE_{AD-CT}}$ & $\mathrm{MHE_{KF}}$ \\
\hline
\\
$N=3$  & 13.22 &  20.44 & 27.37 & 103.42 \\
$N=6$  & 13.22 &  17.77 & 16.25 &  50.68 \\
$N=10$ & 13.22 &  15.28 & 14.51 &  26.54 \\
\hline
\end{tabular}
\end{center}
\caption{Comparison of the effect of the horizon size on the average sum square estimation error for state $x^{(1)}$.}
\label{tab:200trials_x1}
\end{table}

\begin{table}
\begin{center}
\begin{tabular}{ccccc}
\hline
& \multicolumn{4}{c}{$x^{(2)}$} \\

& FIE & $\mathrm{MHE_{AD-VF}}$ & $\mathrm{MHE_{AD-CT}}$ & $\mathrm{MHE_{KF}}$ \\
\hline
\\
$N=3$  & 1.55 & 2.20 & 2.84 & 11.55 \\
$N=6$  & 1.55 & 1.96 & 1.76 & 5.71 \\
$N=10$ & 1.55 & 1.71 & 1.61 & 3.02 \\
\hline
\end{tabular}
\end{center}
\caption{Comparison of the effect of the horizon size on the average sum square estimation error for state $x^{(2)}$.}
\label{tab:200trials_x2}
\end{table}

\subsection{Example 1: $\sigma_w=1.0$, $\sigma_v=0.1$, $N=5$}

In this example, the process noise covariance is set to $\sigma_w=1.0$ and the measurement noise covariance to $\sigma_v=0.1$. Figure \ref{fig:results} shows the states $x^{(1)}$, $x^{(2)}$ and its estimates (see Fig. \ref{fig:x1} and \ref{fig:x2}) with $N=5$. It can be seen that $\mathrm{MHE_{AD}}$ estimation is closer to FIE despite the short horizon chosen. As time passes, the state estimates of $\mathrm{MHE_{AD}}$ converge to the estimates provided by FIE. Figures \ref{fig:x1_err} and \ref{fig:x2_err} show the square error $\xi^{(i)}$ at every sample time. In these figures it can be seen that the error of the proposed methods are similar to the one obtained using FIE. On the other hand, from these figures it can also be seen that the performance of $\mathrm{MHE_{KF}}$ is degraded, mainly because the window size was set to $N=5$. This difference can be explained by the fact that KF update does not approximate properly the arrival cost due to the effect of constraints \cite{robertson2002use}. This situation gets worse when the horizon is shorter, as was mentioned in the previous experiment.

\begin{figure}
	\begin{subfigure}[b]{.45\linewidth}
	\includegraphics[width=1.0\textwidth]{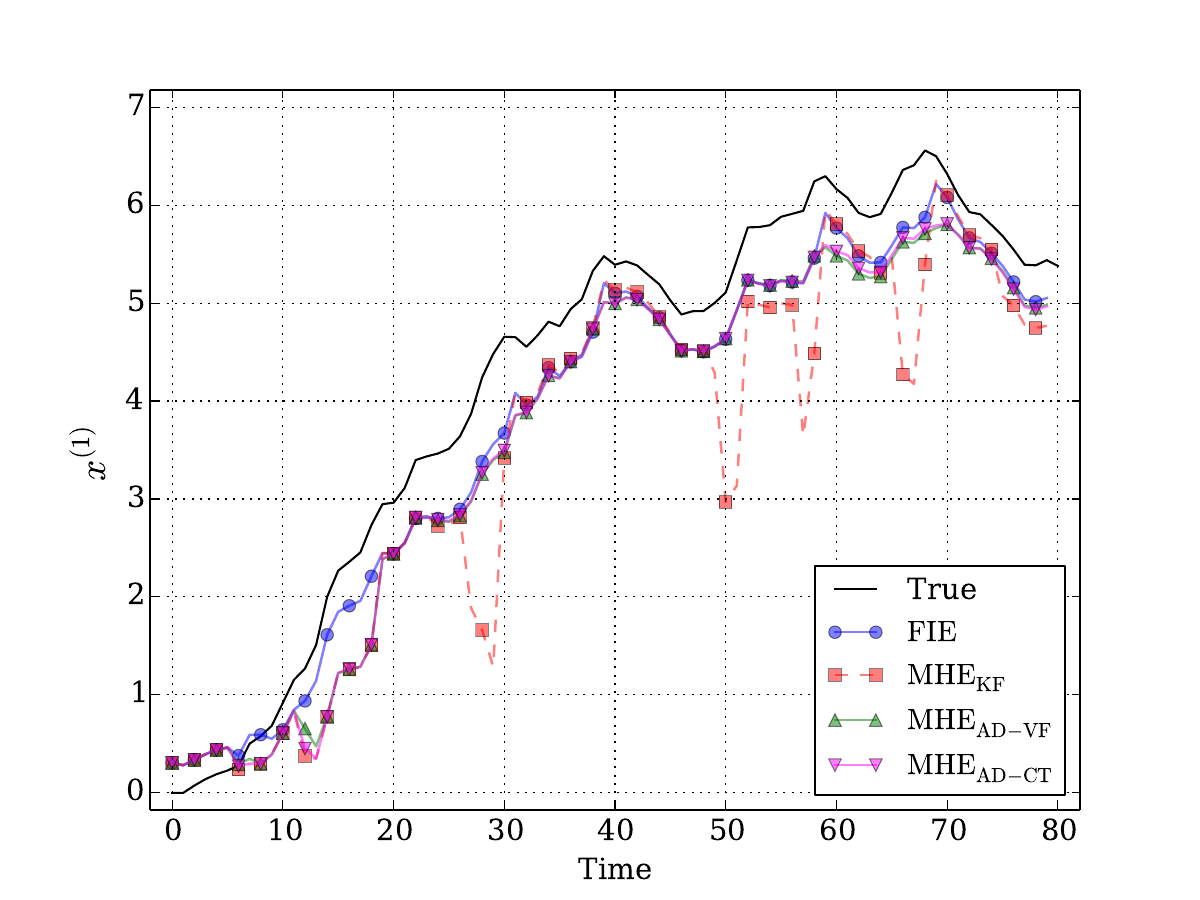}
	\caption{$x^{(1)}$}{}\label{fig:x1}
	\end{subfigure}
	\begin{subfigure}[b]{.45\linewidth}
	\includegraphics[width=1.0\linewidth]{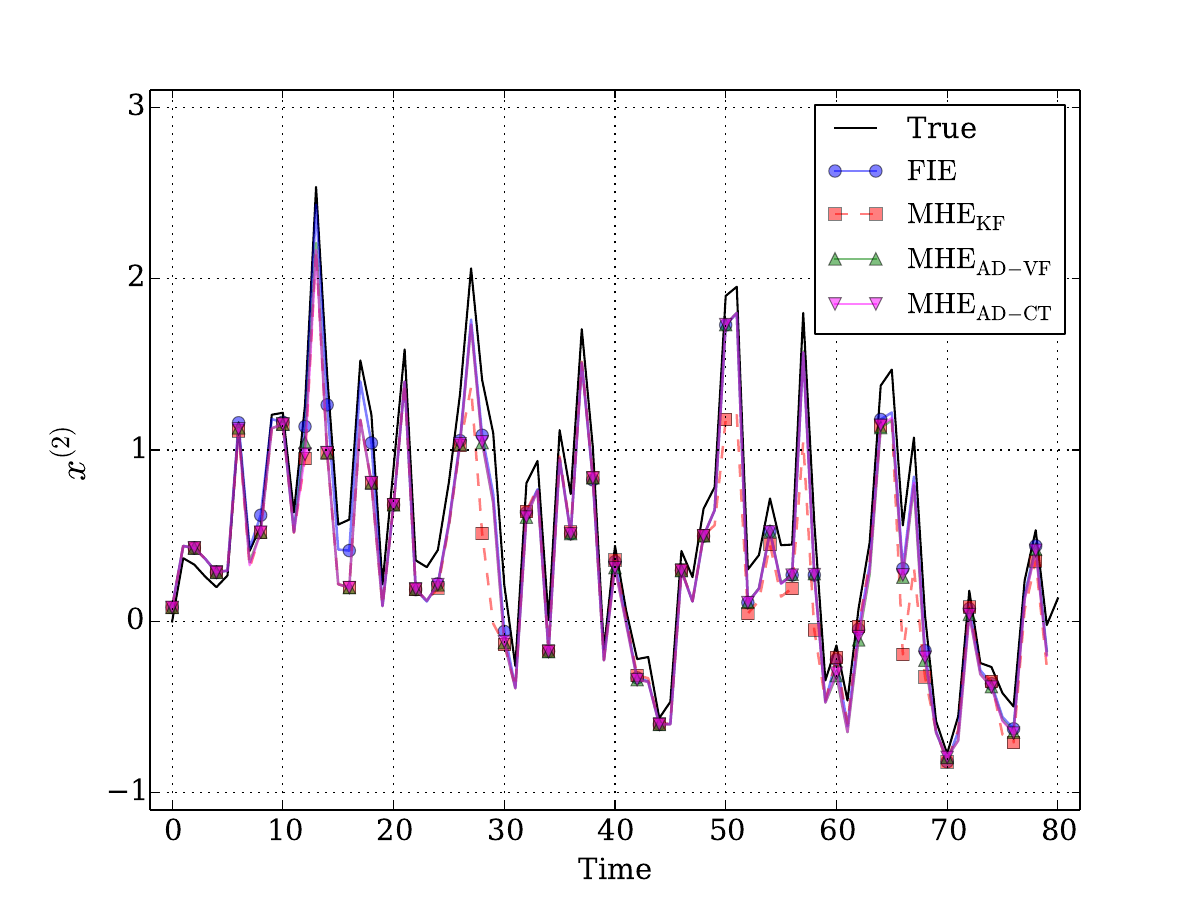}
	\caption{$x^{(2)}$}{}\label{fig:x2}
	\end{subfigure}

	\begin{subfigure}[b]{.45\linewidth}
	\includegraphics[width=1.0\linewidth]{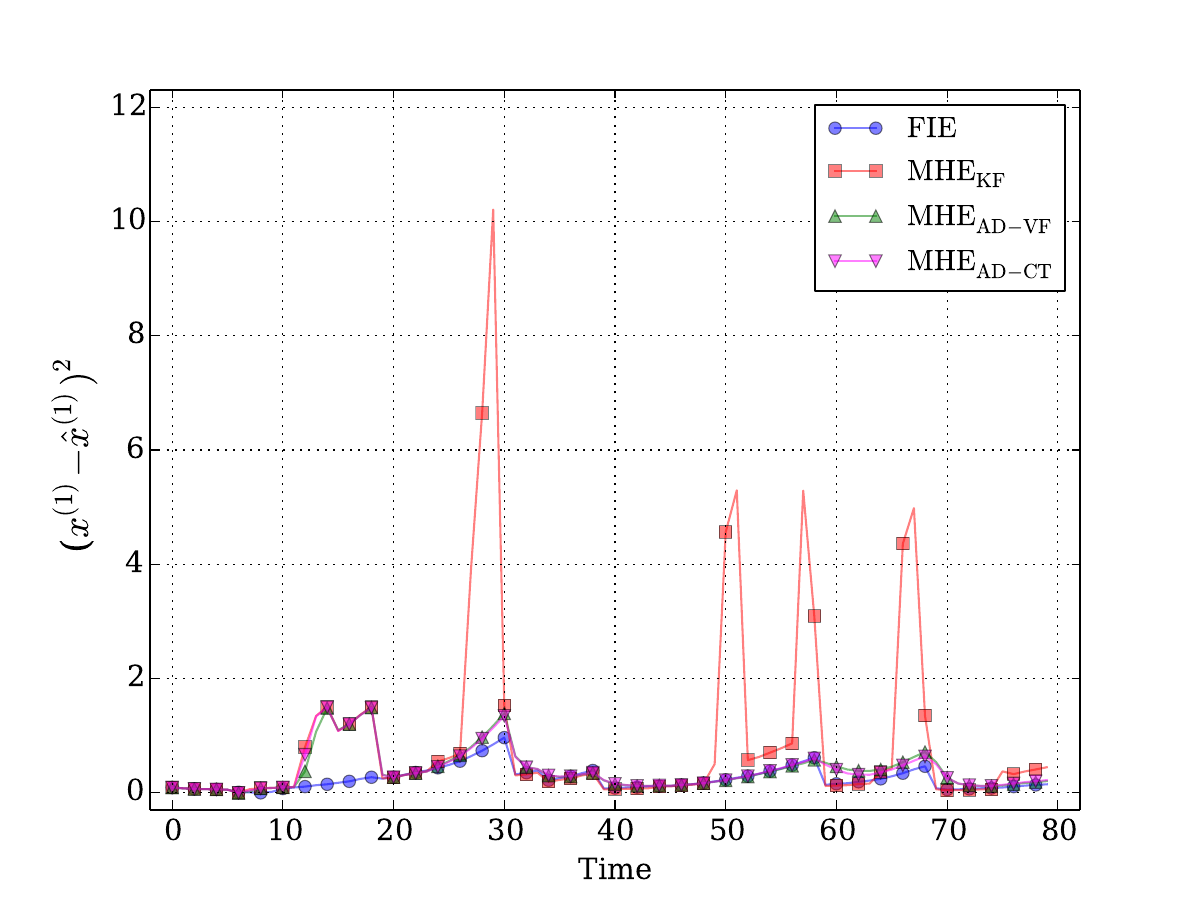}
	\caption{$(x^{(1)}-\hat{x}^{(1)})^2$}{}\label{fig:x1_err}
	\end{subfigure}
	\begin{subfigure}[b]{.45\linewidth}
	\includegraphics[width=1.0\linewidth]{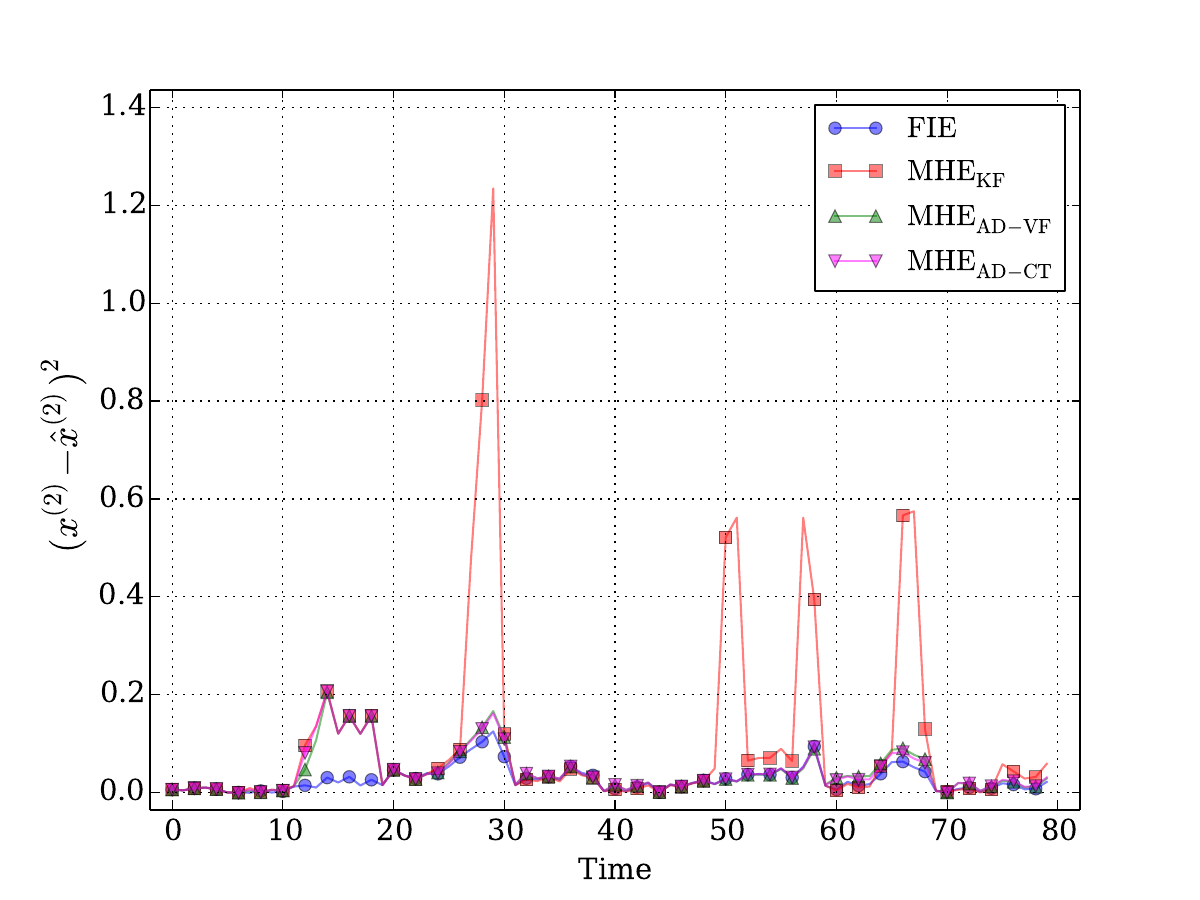}
	\caption{$(x^{(2)}-\hat{x}^{(2)})^2$}{}\label{fig:x2_err}
	\end{subfigure}
	\caption{Comparison of estimators for model (\ref{results_model}) with horizon size $N=5$.}
	\label{fig:results}
\end{figure}

\subsection{Example 2: $\sigma_w=0.4$, $\sigma_v=0.1$, $N=5$}

In this example, although the process noise covariance is set to $\sigma_w=0.4$, it varies with time and the different MHE algorithms do not have knowledge of this happening.  The covariance $\sigma_w$ will take values $0.1$ when $0\leq k <20$, $0.25$ when $20 \leq k < 40$, $1.5$ when $40 \leq k < 100$, $1.0$ when $100 \leq k < 150$ and $0.4$ when $150 \leq k < 200$. We can see in Fig. \ref{fig:x1_x2_varvar} that the proposed algorithms still perform better than $\mathrm{MHE}_{KF}$. Fig. \ref{fig:pinv_varvar} shows the trace of the matrix $P_{k-N}^{-1}$ used by every algorithm except for FIE, since in that case the matrix $P_0$ does not evolve with time. We can see that the $\mathrm{MHE_{KF}}$ matrix converges to a steady state value and never changes again, while the proposed algorithms adapt its weight according to the changing conditions of the process noise. It can be seen that the trace of $P_{k-N}^{-1}$ grows when the error is smaller, showing that it has more confidence on the estimates. On the other hand, when the difference between the real state and the estimate grows, the trace of the weighting matrix becomes smaller, providing better adaptation even under the influence of unknown and varying disturbances.
\begin{figure}
	\begin{subfigure}[b]{.45\linewidth}
	\includegraphics[width=1.0\linewidth]{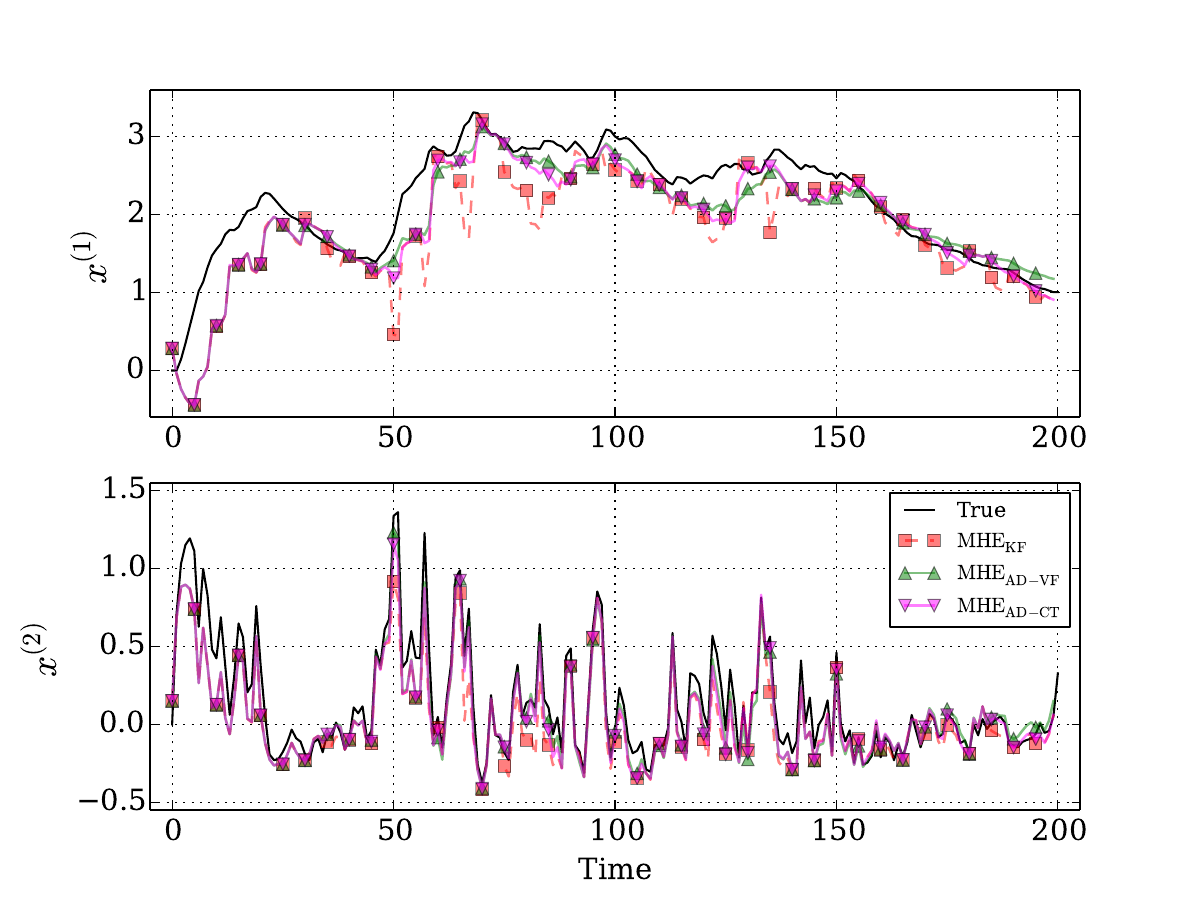}
	\caption{$x^{(1)}$ and $x^{(2)}$}{}\label{fig:x1_x2_varvar}
	\end{subfigure}
	\begin{subfigure}[b]{.45\linewidth}
	\includegraphics[width=1.0\linewidth]{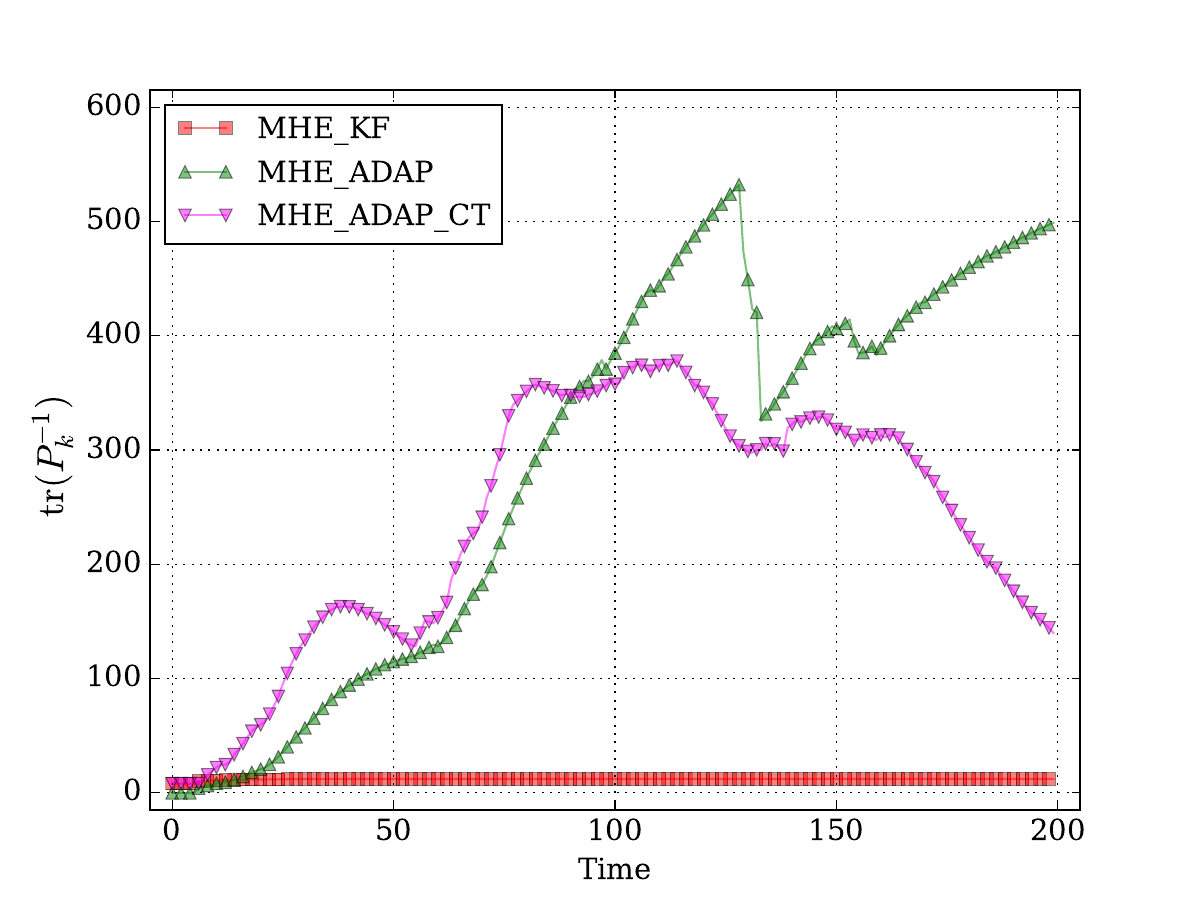}
	\caption{Evolution of $\mathrm{trace}(P^{-1}_k)$}{}\label{fig:pinv_varvar}
	\end{subfigure}
	\caption{Comparison of estimators for model (\ref{results_model}) with horizon size $N=5$ and variable $\sigma_w$.}
	\label{fig:results_varvar}
\end{figure}
In Figures \ref{fig:fig_ff_vf_n_5} and \ref{fig:fig_ff_ct_n_5} we can see the evolution of the forgetting factors of $\mathrm{MHE_{VF}}$ and $\mathrm{MHE_{CT}}$, respectively. In both cases we can see that $\alpha_k$ is nearly $1$ and the algorithm behaves like a weighted least-squares estimator with very long asymptotic memory length. In the cases where the noise increases and the estimator is not following the system, the forgetting factor decreases in order to allow for rapid adaptation of the weighting matrix. This is the main reason why the proposed algorithm provides a better approximation of the arrival cost.
\begin{figure}
	\begin{subfigure}[b]{.45\linewidth}
	\includegraphics[width=1.0\linewidth]{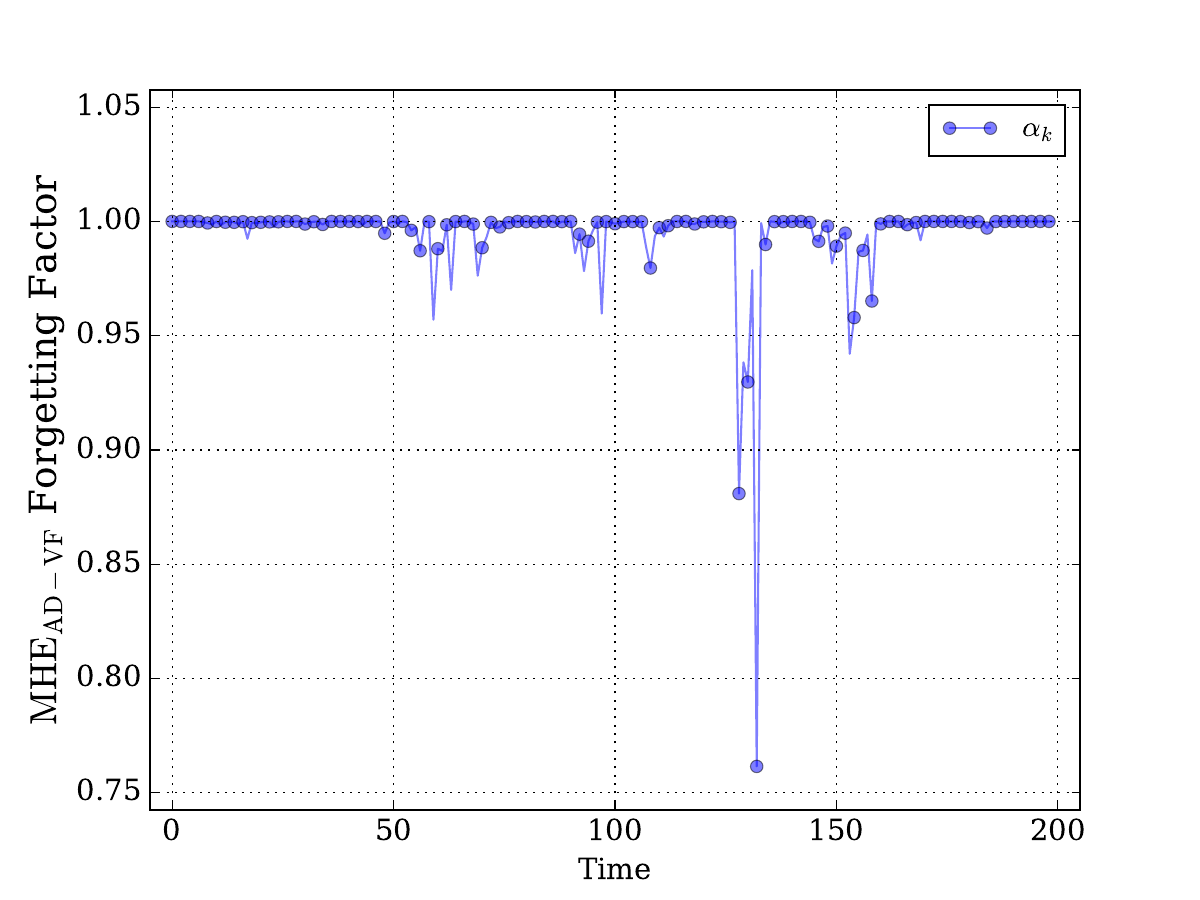}
	\caption{}{}\label{fig:fig_ff_vf_n_5}
	\end{subfigure}
	\begin{subfigure}[b]{.45\linewidth}
	\includegraphics[width=1.0\linewidth]{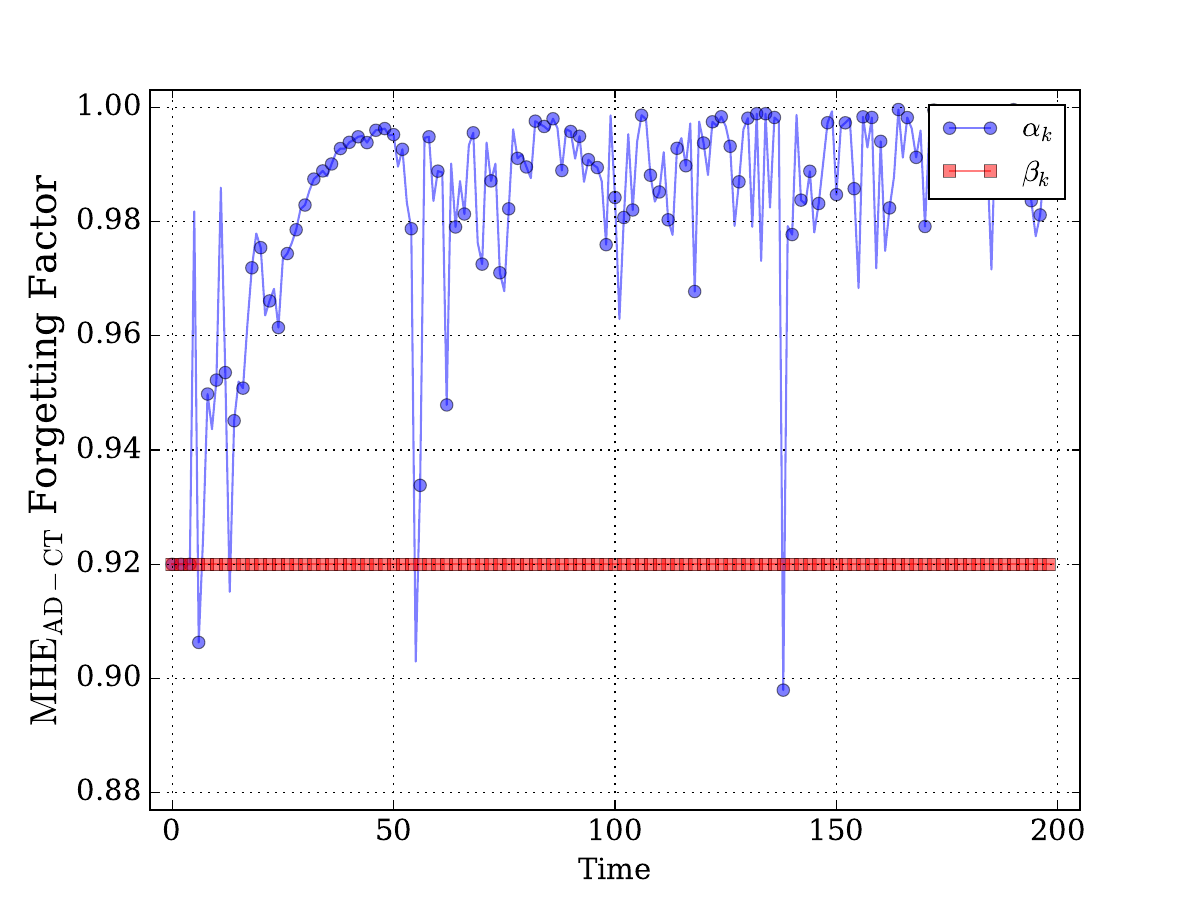}
	\caption{}{}\label{fig:fig_ff_ct_n_5}
	\end{subfigure}
	\caption{Evolution of the sequences of forgetting factors for $\mathrm{MHE_{VF}}$ and $\mathrm{MHE_{CT}}$.}
	\label{fig:results_ff}
\end{figure}

\section{Conclusion}
\label{sec_concl}

In this work we investigated an alternative method to approximate the arrival cost for the MHE problem on linear systems. To the best of our knowledge, adaptive estimation techniques have not been used before to update the weighting matrix $P_{k-N}$. From the results obtained we could show that adaptive estimation techniques give a good approximation of the arrival cost.  This fact is due to the ability of the variable forgetting factor to provide good adaptation capabilities allowing the estimator to incorporate the relevant information of the data and follow not only slow changes but also sudden changes in the dynamics. It is worth mentioning that the proposed approach can be easily extended to nonlinear systems, since the update mechanism only depends on the state estimates and not on the model. The use of CasADi toolbox \cite{Andersson2013b} allows for fast implementation of the optimization algorithms, leading to real-time operation.

\section{Acknowledgments}
The authors wish to thank Dr. Leandro Di Persia for the insightful comments on covariance estimation. The research was supported by \emph{Universidad Nacional de Litoral} (with CAID 501201101 00529 LI) and \emph{Consejo Nacional de Investigaciones Cient\'ificas y T\'ecnicas} (CONICET) from Argentina.

\printbibliography

\end{document}